\documentclass{article}
\usepackage{arxiv}

\usepackage[utf8]{inputenc}
\usepackage[T1]{fontenc}
\usepackage{hyperref}
\usepackage{amsmath,amssymb,mathtools}
\usepackage{mathrsfs}
\usepackage{algorithm}
\usepackage{algorithmic}
\usepackage{booktabs}
\usepackage{amsthm}
\usepackage{xcolor}

\usepackage{authblk}

\newtheorem{theorem}{Theorem}
\newtheorem{lemma}{Lemma}

\newtheorem{definition}{Definition}
\newtheorem{remark}{Remark}

\DeclareMathOperator{\sat}{sat}
\newcommand{\aiping}[1]{\textcolor{black}{#1}}

\usepackage{etoolbox} 
\newenvironment{prettykeywords}{%
  \par\vspace{0.4em}%
  \noindent\small 
  \begin{minipage}{0.97\textwidth}%
    \setlength{\parindent}{0pt}%
    \raggedright
    \hangindent=2em \hangafter=1 
    \textbf{\textit{Keywords}}\textit{:}\ %
}{%
  \end{minipage}\par\vspace{0.6em}%
}

\date{}
\makeatletter
\renewcommand{\@date}{}
\makeatother

\title{Adaptive Event-Triggered MPC for Linear Parameter-Varying Systems with State Delays, Actuator Saturation and Disturbances}

\author[a]{Aiping Zhong}
\author[b]{Wanlin Lu}
\author[a,*]{Langwen Zhang}
\author[a]{Ziyang Bao}

\affil[a]{School of Automation Science and Engineering, South China University of Technology, Tianhe, Guangzhou, 510641, Guangdong, China}
\affil[b]{School of Engineering, The Hong Kong University of Science and Technology, Hong Kong, China}

\affil[*]{Corresponding author: \texttt{aulwzhang@scut.edu.cn}}

\begin{document}
\twocolumn[
\maketitle
\begin{abstract}
This paper proposes a unified adaptive event-triggered model predictive control (ETMPC) scheme for linear parameter-varying (LPV) systems subject to state delays, actuator saturation, and external disturbances. In existing studies, only a limited number of ETMPC methods have attempted to address either state delays or actuator saturation, and even these few methods typically lack co-design optimization between adaptive event-triggering mechanisms and the control law. To overcome these limitations, this paper presents a \aiping{Lyapunov-Krasovskii}-based adaptive ETMPC strategy that enables the co-design optimization of both the triggering mechanism and the controller. Specifically, the event-triggering parameter matrix is adaptively optimized by embedding an internal adaptive variable within the \aiping{Lyapunov-Krasovskii}-like function. Furthermore, the actuator saturation nonlinearity is transformed into a convex hull representation. The infinite-horizon robust optimization problem is reformulated as a convex optimization problem with linear matrix inequality (LMI) constraints. Invariant set constraints are introduced to ensure recursive feasibility, and mean-square input-to-state stability (ISS) under multiple uncertainties is rigorously established. Simulations on an industrial electric heating system validate the proposed method's effectiveness in reducing communication load.
\end{abstract}
\begin{prettykeywords}
Model predictive control; adaptive event-triggering mechanism; linear parameter-varying; state delays; actuator saturation
\end{prettykeywords}
]



\section{Introduction}
Model predictive control (MPC) is an advanced control strategy that generates control actions by iteratively solving a optimal control problem (OCP) online. Its capability to explicitly handle multiple constraints has made it attractive.
However, in practical applications, systems are inevitably affected by parameter uncertainties. 
Specifically, in many industrial systems, the dynamic often varies with operating conditions such as temperature or speed. To capture this variability, linear parameter-varying (LPV) systems with polytopic uncertainties provide an effective modeling framework~\cite{han2021robust}. The key advantage of LPV models is that nonlinear dynamics can be represented as a collection of linear models whose parameters vary with measurable scheduling variables, without requiring system linearization~\cite{sezame2013robust}. Meanwhile, time delays caused by network transmission or physical transport introduce fundamental challenges to the design of MPC~\cite{wang2025advancing}. The presence of such delays implies that the system's future evolution depends not only on its current state but also on a sequence of past states, which significantly increases the complexity of designing control for LPV systems.

Still, the control of time-delay systems becomes significantly more challenging in the presence of actuator saturation, which is a common phenomenon in practical control applications~\cite{desouza2018iss}. Actuator  saturation can limit control signal transmission, potentially degrading system performance and even endangering closed-loop stability~\cite{cao2005minmax}. Within the optimization framework of MPC, an effective strategy to handle this nonlinear constraint is to represent the saturated input as a convex combination of the vertices of the admissible input set~\cite{cao2005minmax, zhang2021distributed, zhang2021event, wu2024robust}. While this transformation is elegant, it substantially increases the number of linear matrix inequality (LMI) constraints. This, in turn, imposes a significant computational burden on the online implementation of the MPC controller.

To reduce the computational burden of online optimization in MPC, we introduce an adaptive event-triggered mechanism (ETM) into the MPC design, resulting in the so-called event-triggered MPC (ETMPC). The adaptive ETM improves the decision-making process of MPC updates by incorporating an internal dynamic variable, allowing for more intelligent and smoother triggering.
In recent years, various ETM strategies have been integrated into MPC frameworks, including static ETMPC~\cite{zhang2021event, he2024pid, zou2019event, hashimoto2025learning}, adaptive ETMPC~\cite{zhan2019adaptive, wang2022robust, li2023adaptive, yuan2023adaptive}, compound ETMPC~\cite{kang2022compound}, and periodic ETMPC~\cite{wang2023periodic,wang2023periodic2, wang2024periodic, resmi2023distributed}. Among them, static ETMPC is typically categorized into approaches based on state error~\cite{zhang2021event, he2024pid} and those based on cost function error~\cite{zou2019event, hashimoto2025learning}. Adaptive ETMPC~\cite{zhan2019adaptive, wang2022robust, li2023adaptive, yuan2023adaptive} extends static ETMPC by enabling adaptive adjustment of the triggering threshold. Compound ETMPC~\cite{kang2022compound} further enhances static or adaptive ETMPC by incorporating a minimum inter-event time condition. Periodic ETMPC~\cite{wang2023periodic,wang2023periodic2, wang2024periodic, resmi2023distributed} explores ETMPC strategies under periodic monitoring schemes. 

Despite notable advances in ETMPC, few studies address its integration with actuator saturation or time-delay systems, and existing works still show certain limitations. For example,~\cite{zhang2021event}, the only work combining ETMPC with actuator saturation, has limitations: 1) its static ETM with fixed threshold lacks adaptability for complex dynamics; 2) it does not account for external disturbances. This omission is significant, as robustness against persistent disturbances is a crucial benchmark for any ETM strategy, and all other ETMPC studies cited in this paper have considered the disturbances. 
Similarly, the only study in~\cite{li2023adaptive}, combining ETMPC with time-delay systems, has shortcomings: 1) its adaptive ETM, relying solely on the most recent triggering information, is memoryless; 2) its sign function-based switching logic risks threshold chattering, especially near marginal stability; 3) the triggering condition is treated as an external constraint, preventing co-design with the controller. This limitation confines the framework to guaranteeing only $H_\infty$ performance. In contrast, input-to-state stability (ISS) offers a deeper and more comprehensive guarantee of robustness under persistent disturbances.

In conclusion, there remains a critical need for a unified control framework to handle LPV systems with polytopic uncertainties, state delays, and actuator saturation. Such a framework should integrate adaptive ETM into robust MPC design and provide rigorous ISS guarantees via an integrated co-design approach. To address this gap, this paper proposes a novel adaptive ETMPC scheme with the following key contributions:
\begin{itemize}
    \item A unified adaptive ETMPC framework is proposed for LPV systems with disturbances, state delays, and actuator saturation. An adaptive ETM is presented for MPC, wherein the triggering threshold is dynamically adjusted via an internal adaptive variable.
    \item To systematically address these coupled challenges, a co-design method is developed by embedding the ETM into a tailored \aiping{Lyapunov-Krasovskii}-like function, thereby enabling the joint optimization of the controller{,} auxiliary matrices for saturation, and triggering parameters.
    \item Invariant set constraints are introduced to guarantee recursive feasibility of the online optimization problem. Furthermore, it is proved that the closed-loop system achieves mean-square ISS under multiple uncertainties.
\end{itemize}

The remainder of this paper is organized as follows. Section~\ref{Sec2} presents the considered system model, introduces the adaptive ETM and {formulates} the robust {MPC} problem. Section~\ref{Sec3} provides the complete design procedure of the adaptive ETMPC controller. Section~\ref{Sec4} demonstrates the effectiveness of the proposed method through a simulation case study on an industrial electric heater system.

\textit{Notation:} The set of real numbers is denoted by $\mathbb{R}$. The symbols $\mathbb{R}^{n_x}$ and $\mathbb{R}^{n_x \times n_u}$ represent the set of $n_x$-dimensional real vectors and $n_x \times n_u$ real matrices, respectively. The identity matrix of appropriate dimensions is denoted by $I$. For a vector $x$ and a positive definite matrix $M$, the norm $\|x\|_M$ is defined as $\sqrt{x^\top M x}$. The operator $\operatorname{blkdiag}\{\cdot\}$ forms a block-diagonal matrix. The convex hull of a set of matrices is denoted by $\text{Co}\{\cdot\}$. An asterisk ($\ast$) in a symmetric matrix represents a term that is induced by symmetry.

\vspace{-2ex}
\section{Preliminaries and Problem Formulations }\label{Sec2}
In this section, we first present the LPV system dynamics. Next, we introduce an adaptive event-triggered mechanism along with a min-max MPC formulation. Finally, we summarize the key challenges that must be addressed in the subsequent controller design.
\subsection{Disturbed LPV system dynamics with delays and saturations}
In this work, we consider a global polytopic uncertain discrete-time LPV system with state delays, actual saturation and external disturbances,  described as
\begin{equation}
x_{k+1}=A(k)x_k+\sum_{\rho=1}^l\tilde{A}_\rho(k)x_{k-\tau_\rho}+B(k)\sigma (u_k)+D\omega_k.
\label{model1}
\end{equation}
The vectors $x_k$ and $x_{k-\tau_\rho} \in \mathbb{R}^{n_x}$ denote the current and delayed system states, respectively, where $\tau_\rho$ is a known integer representing the state delay, with $\rho = 1, \ldots, l$. The matrices $A(k)$ and $\tilde{A}_\rho(k) \in \mathbb{R}^{n_x \times n_x}$ are the system and delayed dynamics matrices, respectively. The vector $\sigma(u_k) \in \mathbb{R}^{n_u}$ represents the saturated input, and $B(k) \in \mathbb{R}^{n_x \times n_u}$ is the input matrix. The saturated input $\sigma(u_k) \in \mathbb{R}^{n_u}$ is defined component-wise as $\sigma_i(u_i(k)) = \operatorname{sign}(u_i(k)) \cdot \min\{u_{\sat}, |u_i(k)|\}$, constraining each control input within $[-u_{\sat}, u_{\sat}]$.
The vector $\omega_k \in \mathbb{R}^{n_\omega}$ denotes a persistent disturbance, while $D \in \mathbb{R}^{n_x \times n_\omega}$ is a constant disturbance matrix. Bounded disturbance $\omega_k$ satisfies 
\begin{equation}
    \omega_k \in \mathcal{W}\overset{def}{=} \left\{\omega|\omega^\top \omega\le d^2\right\}, \label{disturbance_bound}
\end{equation}
where $d > 0$ is a known constant.

It is assumed that the system matrices are fully known and can be expressed as affine functions of a time-varying parameter vector $\alpha(k) = [\alpha_1(k), \alpha_2(k), \dots, \alpha_L(k)]$. Specifically, the matrices satisfy
\begin{equation}
\begin{aligned}
    &[A(k),\ \tilde{A}_1(k),\ \cdots,\ \tilde{A}_l(k),\ B(k)] \\
    &= \sum_{v=1}^{L} \alpha_v(k) 
    [A^v,\ \tilde{A}_1^v,\ \cdots,\ \tilde{A}_l^v,\ B^v],
\end{aligned}
\label{para_vary}
\end{equation}
where each weighting coefficient $\alpha_v(k)$ satisfies $0 \le \alpha_v(k) \le 1$ and $\sum_{v=1}^{L} \alpha_v(k) = 1$. The parameter-dependent system matrices belong to a known convex polytope $\mathcal{D}$, which is spanned by $L$ vertex matrices as
\begin{equation}
\begin{aligned}
&\big[A(k),\,\tilde A_1(k),\,\ldots,\,\tilde A_\ell(k),\,B(k)\big] \in \mathcal D,\\
&\mathcal D \coloneqq \operatorname{co}\Big\{ \big[A^1,\tilde A_1^1,\ldots,\tilde A_\ell^1,B^1\big],\,\ldots,\\
&\qquad\qquad \big[A^L,\tilde A_1^L,\ldots,\tilde A_\ell^L,B^L\big] \Big\}.
\end{aligned}
\label{convex_hull}
\end{equation}

In addition, we consider a linear state feedback law given by $u_k = F_k x_k$. To facilitate the analysis under actuator saturation, the saturated feedback $\sigma(u_k)$ can be represented as an element within a convex combination of a set of auxiliary linear feedback laws. This reformulation relies on the following two lemmas.
\begin{lemma}[\cite{hindi1998analysis}, Lemma~$1$]    \label{Lemma_1}
    Given $F \in \mathbb{R}^{n_u \times n_x}$, define the set $ \psi(F) = \left\{ x \in \mathbb{R}^{n_x} \,\middle|\, |f_i^\top  x| \le u_{\sat},i = 1, \ldots, n_u \right\}$, where $f_i^\top $ denotes the $i^{th}$ row of $F$. Let $P \in \mathbb{R}^{n_x \times n_x}$ be a positive definite matrix and $\gamma > 0$ be a scalar. Then, the ellipsoidal set  $\Omega(P, \gamma) \coloneqq \left\{ x \in \mathbb{R}^{n_x} \,\middle|\, x^\top  P x \le \gamma \right\}$ is contained within $\psi(F)$ if and only if     \begin{equation}
        f_i^\top  (\frac{P}{\gamma})^{-1} f_i \le u_{\sat}^2,\quad i = 1, \ldots, n_u. \label{sat_lem_eq}
    \end{equation}
\end{lemma}
\begin{lemma}[\cite{hu2002analysis}]\label{Lemma_2}
    Let $F, H \in \mathbb{R}^{n_u \times n_x}$ be two given matrices. Suppose the condition     \begin{equation}
            |h_i^\top  x| \le u_{\sat}  \label{H_LMI_orignal}
    \end{equation}
holds for all $i = 1, \ldots, n_u$, where $h_i^\top $ denotes the $i^{th}$ row of~$H$. Then, the saturated control input $\sigma(Fx)$ can be represented as a term belonging to a convex hull of auxiliary linear feedback, given by
    \begin{equation}
        \sigma(Fx) \in \text{Co} \left\{ E_\eta Fx + E_\eta^{-} Hx \;\middle|\; \eta = 1, \ldots, 2^{n_u} \right\}, 
    \end{equation}
For each $\eta$, the matrix $E_\eta \in \mathbb{R}^{n_u \times n_u}$ is a diagonal matrix with diagonal entries either $0$ or $1$, and $E_\eta^{-} = I - E_\eta$.
\end{lemma}

Based on Lemma~\ref{Lemma_2}, the saturated $\sigma(u_k)$ is given by
\begin{equation}
   \sigma(u_k) = \sigma(F_kx_k) = \sum_{\eta = 1}^{2^{n_u}} \varrho_\eta ( E_{\eta,k} F_kx_k + E_{\eta,k}^{-} H_kx_k ),
    \label{saturation}
\end{equation}
where $0 \le \varrho_\eta \le 1$ and $\sum_{\eta = 1}^{2^{n_u}} \varrho_\eta = 1$ with $[E_{\eta,k}, E_{\eta,k}^{-}] \in \Pi$, $\Pi\coloneqq\text{Co}\left\{ [E_1, E_1^{-}], \ldots, [E_{2^{n_u}}, E_{2^{n_u}}^{-}] \right\}.$

\subsection{Adaptive event-triggered mechanism}
In order to effectively reduce the frequency of data packet transmissions to alleviate the communication burden, we first introduce an adaptive event-triggered {mechanism}. Specifically, let $\mathcal{S} = \{k_t \mid t \in \mathbb{N}\}$ denote the set of event-triggering instants, which collects all triggering time points in chronological order. 
Suppose an event is triggered at time $k_t$. Then, due to the zero-order hold (ZOH) mechanism, the control input is unchanged until the next triggering instant $k_{t+1}$, that is,
\[
u_k = u_{k_t}, \quad \forall\, k \in \{k_t, k_t+1, \ldots, k_{t+1}-1\}.
\]
The next triggering instant $k_{t+1}$ is determined by the following rule
\begin{equation}
    k_{t+1} \!=\! k_t + \underset{r \in \mathbb{N}}{\min}
    \left\{\! r \middle|\, \varepsilon ( \left\|e_{k_t+r} \right\|^2_{\Phi_{k_t}}\!\! - \theta \left\|x_{k_t+r}\right\|^2_{\Phi_{k_t}} ) \!> \!\beta_{k_t+r} \right\}, \label{adaptiveETM}
\end{equation}
where the error vector $e_k$ is defined as
\begin{equation}
    e_k \coloneqq x_{k_t} - x_k, \quad k \in \{k_t, k_t + 1, \ldots, k_{t+1} - 1\}, \label{error}
\end{equation}
which depends on the deviation of the current state $x_k$ from the state $x_{k_t}$ at the last triggering instant. The matrix ${\Phi_{k_t}} \in \mathbb{R}^{n_x \times n_x}$, called event-triggering parameter matrix, is a given positive definite matrix. The scalar $\beta_k$ represents an internal adaptive variable governed by the recursive update
\begin{equation}
    \beta_{k+1} = \mu \beta_k - ( \left\|e_{k} \right\|^2_{\Phi_{k_t}} - \theta \left\|x_{k}\right\|^2_{\Phi_{k_t}} ), \beta_0 \ge 0, \label{internal_adaptive_variable}
\end{equation}
where $\mu$, $\theta$, and $\varepsilon$ are given constants satisfying
\begin{equation}
    0 < \mu < 1,0 < \theta < 1,\varepsilon \ge 1/\mu. \label{para_constraints}
\end{equation}
According to the rule in~\eqref{adaptiveETM}, for all $k \in \mathbb{N} \setminus \mathcal{S}$ (i.e., during non-triggering intervals), the following inequality holds
\begin{equation}
    \varepsilon ( \left\|e_{k} \right\|^2_{\Phi_{k_t}} - \theta \left\|x_{k}\right\|^2_{\Phi_{k_t}} ) \le \beta_k. \label{time_sequnence2}
\end{equation}
Substituting~\eqref{time_sequnence2} into~\eqref{internal_adaptive_variable}, we obtain
\[
\beta_{k+1} \ge \mu \beta_k - \frac{1}{\varepsilon} \beta_k = ( \mu - \frac{1}{\varepsilon} ) \beta_k,
\]
which leads to the recursive lower bound
\[
\beta_k \ge ( \mu - \frac{1}{\varepsilon} )^k \beta_0.
\]
Since $\mu \ge 1/\varepsilon$ by assumption, it follows that $\beta_k \ge 0$ for all $k \in \mathbb{N}$.
\begin{remark}
The triggering condition in~\eqref{adaptiveETM} includes an internal adaptive variable $\beta_k$, corresponding to the adaptive ETM proposed in~\cite{Hu2016}. When the influence of $\beta_k$ becomes negligible (i.e., $\varepsilon \to \infty$), it reduces to the static ETM in \cite{Hu2016}, with a fixed zero-based threshold 
\begin{equation} 
    k_{t+1} = k_t + \underset{r \in \mathbb{N}}{\min} \left\{ r \,\middle|\, \left\|e_{k_t+r} \right\|^2_{\Phi_{k_t}} - \theta \left\|x_{k_t+r}\right\|^2_{\Phi_{k_t}} > 0 \right\},
    \label{staticETM}
\end{equation}
The key difference lies in the threshold. Static ETM uses a fixed value, while adaptive ETM employs a adaptive threshold $\beta_k/\varepsilon$, ensuring less frequent triggering and reduced {communication burden.} This adaptive ETM was integrated with MPC in~\cite{shi2021dynamic} for Markov jump systems.
\end{remark}

Based on~\eqref{saturation} and~\eqref{error}, the controller can be rewritten as
\begin{equation}
    \sigma(u_k) = \sum_{\eta = 1}^{2^{n_u}} \varrho_\eta ( E_\eta F_{k_t} + E_\eta^{-} H_{k_t} )(x_k + e_k).
    \label{saturation_ET}
\end{equation}
Then the closed-loop systems~\eqref{model1} can be rewritten as
\begin{equation}
\begin{aligned}
    x_{k+1}=&A(k)x_k+\sum_{\rho=1}^l\tilde{A}_\rho(k)x_{k-\tau_\rho}+\\
    &B(k)\sum_{\eta = 1}^{2^{n_u}} \varrho_\eta ( E_\eta F_{k_t} + E_\eta^{-} H_{k_t} )(x_k + e_k)+D\omega_k.
\end{aligned}
\label{model2}
\end{equation}
Due to the presence of persistent external disturbances $\omega_k$ in the system~\eqref{model2}, one of our primary objectives is to ensure that the designed controller guarantees mean-square ISS. The formal definition of mean-square ISS is given below.

\begin{definition}[\cite{song2018n}]\label{Definition_1}
    Define $\tilde{z}_{k} =  [x_{k}^\top,  x_{k-1}^\top,  \dots,  x_{k-\tau_l}^\top]^\top$. Given the initial state $\tilde{z}_0$ and an initial constant $\beta_0$, the closed-loop system of~\eqref{model2} and~\eqref{internal_adaptive_variable} is said to be mean-square ISS if there exist functions $\psi_1 \in \mathcal{KL}$ and $\psi_2 \in \mathcal{K}$ such that
    \begin{equation*}
        \Vert \tilde{z}_{k} \Vert^2_2 \leq \psi_1(\Vert \tilde{z}_{0} \Vert^2_2 + \beta_0,\, k) + \psi_2(\Vert \omega_k \Vert^2_\infty)
    \end{equation*}
holds for all $k \geq 0$.
\end{definition}

\subsection{MPC formulation}
After introducing the system dynamics and the adaptive event-triggered mechanism, we now formulate the MPC strategy. To handle the persistent disturbance $\omega_k$ in the system~\eqref{model2} while ensuring robustness, we adopt a robust MPC framework that minimizes the worst-case cost over an infinite horizon.
The robust MPC problem is formulated as follows
\begin{equation}
    \begin{aligned}
        &\min_{\{u_{k+h|k}\}_{h=0}^{\infty}} \; \max_{\substack{[A(\cdot), \tilde{A}_1(\cdot), \ldots, B(\cdot)] \in \mathcal{D}, \\ \omega(\cdot) \in \mathcal{W},\; \varrho_\eta(\cdot) \in [0,1], \sum\varrho_\eta=1}} J_\infty(k), \\
        &\text{subject to}~\eqref{model1}\text{ and}~\eqref{H_LMI_orignal},
    \end{aligned}\label{cost}
\end{equation}
where the min-max cost function is defined as
\begin{equation}
    J_\infty(k) = \sum_{h=0}^{\infty} ( \|x_{k+h|k}\|_Q^2 + \|u_{k+h|k}\|_R^2 - \varphi \|\omega_{k+h|k}\|_2^2 ). \label{cost_J}
\end{equation}
This cost function is designed to account for both system performance and robustness. The term $\varphi \|\omega\|^2$ penalizes the worst-case impact of persistent disturbances with a given constant $\varphi>0$, allowing the controller to hedge against uncertainties~\cite{teng2018robust}. Here, $Q \in \mathbb{R}^{n_x \times n_x}$ and $R \in \mathbb{R}^{n_u \times n_u}$ are given positive-definite weighting matrices, used to penalize deviations in the state and control input, respectively.

\begin{figure}
    \centering
    \includegraphics[width=0.25\textwidth]{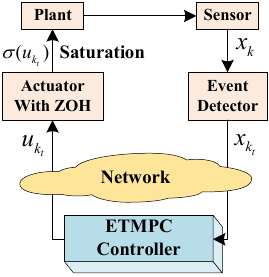}  
    \caption{Architecture of the proposed ETMPC control framework. The structure includes a physical plant~\eqref{model1} with state delays, disturbances and actuator saturation, an event-triggering mechanism based on~\eqref{adaptiveETM}, an ETMPC, and a ZOH actuator.}
    \label{fig:framework}
\end{figure}

The proposed control framework for the system in~\eqref{model2} is illustrated in Fig.~\ref{fig:framework}. The closed-loop architecture consists of a physical plant in~\eqref{model1}, an event detector for~\eqref{adaptiveETM}, an ETMPC controller, and a ZOH actuator with saturation. At each time instant, the {event detector} monitors the system state $x_k$ and decides whether to trigger a transmission. At each triggering instant, the ETMPC controller solves an optimization problem to compute a new control sequence. The resulting control input is then transmitted to the actuator. Otherwise, both control computation and communication are suspended. The ZOH holds this input constant until the next triggering event. The physical actuator implements this command subject to its saturation limits.

\vspace{-1ex}
\subsection{Problem formulations}
Based on the system model and the control framework above, the design of the adaptive ETMPC scheme gives rise to the following key questions:
\begin{enumerate}
    \item How can the min-max optimization problem in~\eqref{cost}, which cannot be solved directly due to its infinite horizon and the combined uncertainties of polytopic uncertainty in~\eqref{para_vary}, external disturbances $\omega$, and saturation $\sigma(\cdot)$, be reformulated into a tractable set of LMIs? \label{prb1}

    \item How can a unified \aiping{Lyapunov-Krasovskii}-like function be constructed to simultaneously stabilize the system in the presence of state delays and embed the ETM's internal adaptive variable $\beta_k$? Furthermore, how can this function serve as the basis for a true co-design framework that jointly optimizes the control gain $F_{k_t}$ and the triggering matrix $\Phi_{k_t}$? \label{prb2}

    \item How can we rigorously guarantee the performance of the resulting closed-loop system in~\eqref{model2} by proving recursive feasibility and mean-square ISS as defined in Definition~\ref{Definition_1}? \label{prb3}
\end{enumerate}

We address the above questions in the subsequent sections. Section~\ref{Sec3.1} focuses on Problems~\ref{prb1} and~\ref{prb2}; Sections~\ref{Sec3.2} and~\ref{Sec3.3} address Problem~\ref{prb3}; and Section~\ref{Sec4} validates the proposed method via simulation.

\vspace{-2ex}
\section{Adaptive ETMPC Controller Design} \label{Sec3}
In this section, we present the design of the proposed adaptive ETMPC controller. We reformulate the robust control problem into a convex optimization by deriving a set of LMIs. We then prove that the resulting controller guarantees recursive feasibility and mean-square ISS for the closed-loop system.

\vspace{-1ex}
\subsection{Upper bound of infinite-horizon $J_\infty$} \label{Sec3.1}
Since the cost function $J_\infty$ in~\eqref{cost} is defined over an infinite horizon and involves uncertainty caused by the disturbance $\omega_k$, it cannot be solved directly. To address this issue, we define the \aiping{Lyapunov-Krasovskii}-like function as
\begin{equation}
    \begin{aligned}
        V_{k+h|k} \!=& x_{k+h|k}^\top  P_k x_{k+h|k} \!
        +\!\!\! \sum_{\rho=1}^{\tau_1} x_{k+h-\rho|k}^\top  P_{k,\tau_1} x_{k+h-\rho|k} \\
     + \cdots + &\sum_{\rho=\tau_{l-1}+1}^{\tau_l} x_{k+h-\rho|k}^\top  P_{k,\tau_l} x_{k+h-\rho|k}
        + \beta_{k+h|k} \\
        =\; & \tilde{z}_{k+h|k}^\top  \tilde{P}_k \tilde{z}_{k+h|k} + \beta_{k+h|k}, h\ge 0,
    \end{aligned}
    \label{quadratic_function_1}
\end{equation}
where $\tilde{z}_{k+h|k} =  [x_{k+h|k}^\top,  x_{k+h-1|k}^\top,  \dots,  x_{k+h-\tau_l|k}^\top]^\top$ is the augmented state vector and the subscript ``$|k$'' indicates the prediction made at time step $k$. The corresponding block-diagonal matrix $\tilde{P}_k$ is 
\[\tilde{P}_k = \operatorname{blkdiag}\{ P_k, \underbrace{P_{k,\tau_1}, \ldots, P_{k,\tau_1}}_{\tau_1}, \ldots, \underbrace{P_{k,\tau_l}, \ldots, P_{k,\tau_l}}_{(\tau_l - \tau_{l-1})}\}.\]
Unlike traditional \aiping{Lyapunov-Krasovskii} function for LPV systems with state delays~\cite{jeong2005constrained, zhong2025resilient}, which typically consider only the state-dependent term $\tilde{z}_{k+h|k}^\top  \tilde{P}_k \tilde{z}_{k+h|k}$, our formulation in~\eqref{quadratic_function_1} further incorporates the internal adaptive variable $\beta_{k+h|k}$ introduced by the ETM.
We now propose the first condition for the \aiping{Lyapunov-Krasovskii}-like function {in}~\eqref{quadratic_function_1}, given by
\begin{equation}
\begin{aligned}
    V_{k+h+1|k} - V_{k+h|k}
    &\le -( \|x_{k+h|k}\|_Q^2 + \|u_{k+h|k}\|_R^2 \\
    &\quad\;\; - \varphi \|\omega_{k+h|k}\|^2_2 ),
\end{aligned}
\label{condition_1}
\end{equation}
By summing~\eqref{condition_1} over $h = 0$ to $h = \infty$, we obtain
\begin{equation}
    J_\infty(k) \le V_{k|k}.
\end{equation}
Next, we define the second upper bound condition for function $V_{k|k}$ in~\eqref{quadratic_function_1} as 
\begin{equation}
    V_{k|k} \le \gamma_k. \label{condition_2}
\end{equation}
Then, based on the two conditions in~\eqref{condition_1} and~\eqref{condition_2}, we can derive the following inequality
\begin{equation}
    J_\infty(k) \le \gamma_k,
\end{equation}
which transforms the original problem of directly minimizing the cost function $J_\infty(k)$ in~\eqref{cost} into an equivalent problem of minimizing the upper bound $\gamma_k$. 
To simplify the notation without loss of generality in the derivation methodology, we consider the case of a single state delay for the remainder of this section. 
In this case, we define $
\overline{z}_{k+h|k} = [x_{k+h|k}^\top,  x_{k+h-1|k}^\top,  \dots,  x_{k+h-\tau_1|k}^\top]^\top$ and $\overline{P}_k = \operatorname{blkdiag}\{ P_k, P_{k,\tau_1}, \ldots, P_{k,\tau_1}\}$. Then the \aiping{Lyapunov-Krasovskii}-like function can be redefined as
\begin{equation}
    \begin{aligned}
        V_{k+h|k} =\; & x_{k+h|k}^\top  P_k x_{k+h|k} 
        +  \\
        &\sum_{\rho=1}^{\tau_1} x_{k+h-\rho|k}^\top  P_{k,\tau_1} x_{k+h-\rho|k}+ \beta_{k+h|k}\\
        =\; & \overline{z}_{k+h|k}^\top  \overline{P}_k \overline{z}_{k+h|k} + \beta_{k+h|k}.
    \end{aligned}
    \label{quadratic_function_2}
\end{equation}

To enforce both conditions~\eqref{condition_1} and~\eqref{condition_2}, we now propose Lemma~\ref{Lemma_3}. 
\begin{lemma}\label{Lemma_3}
Consider the closed-loop system in~\eqref{model2}. Let $Q$, $R$, and $\varphi$ be given matrices and scalar in~\eqref{cost_J}, respectively, and let $\theta$ be the given scalar in~\eqref{para_constraints}.
If there exist matrices $W_k$, $W_{k,\tau_1}$, $Y_{4,k}$, $Y_{5,k}$, positive definite matrices $Y_{1,k}$, $Y_{2,k}$, $Y_{3,k}$, and a scalar $\gamma_k > 0$, such that the following LMIs 
    \begin{equation}
        \begin{bmatrix}
            \Theta_1 & \ast\\
            \Theta_2 & \Theta_3
        \end{bmatrix} \ge 0, \label{LMI1}
    \end{equation} and 
            		\begin{equation}
			\begin{bmatrix}
				1&\overline{z}&\beta^{0.5}_{k}\\
				
				\ast&\Theta_4&0\\
				
				\ast&\ast&\gamma_k\\
			
			\end{bmatrix}\ge0,\label{LMI2}
		\end{equation}
   hold, where
\begin{align}
    \Theta_1 &= \operatorname{blkdiag} \left\{
        W_k + W_k^\top - Y_{1,k},\;
        W_{k,\tau_1} + W_{k,\tau_1}^\top - Y_{2,k},\right. \nonumber \\
        &\hspace{4em} \left. W_k + W_k^\top - Y_{3,k},\;
        \varphi \gamma_k I
    \right\}, \nonumber \\
    \Theta_2 &= \begin{bmatrix}
        A^v W_k + B^v \mathscr{N}_k & \tilde{A}_1^v W_{k,\tau_1} & B^v \mathscr{N}_k & \gamma_k D \\
        W_k & 0 & 0 & 0 \\
        Q^{0.5} W_k & 0 & 0 & 0 \\
        R^{0.5} \mathscr{N}_k & 0 & R^{0.5} \mathscr{N}_k & 0 \\
        \theta^{0.5} W_k & 0 & 0 & 0
    \end{bmatrix}, \nonumber \\
    \Theta_3 &= \operatorname{blkdiag} \left\{ Y_{1,k},\; Y_{2,k},\; \gamma_k I,\; \gamma_k I,\; Y_{3,k} \right\}, \nonumber \\
    \Theta_4 &= \operatorname{blkdiag} \left\{ Y_{1,k},\; Y_{2,k},\; \dots,\; Y_{2,k} \right\}, \nonumber \\
    \mathscr{N}_k &= E_\eta Y_{4,k} + E_\eta^{-} Y_{5,k}, \quad v = 1, \ldots, L. \label{Nk}
\end{align}
then the conditions~\eqref{condition_1} and~\eqref{condition_2} hold, where
   $ P_k = \gamma_k Y_{1,k}^{-1}, P_{k,\tau_1} = \gamma_k Y_{2,k}^{-1}, {\Phi_k} = \gamma_k Y_{3,k}^{-1},  F_k = Y_{4,k} W_k^{-1}, H_k = Y_{5,k} W_k^{-1}.$
\end{lemma}
\begin{proof}
Substituting the update law of $\beta_k$ from~\eqref{internal_adaptive_variable}, and noting that $\beta_k \ge 0$ for all $k \in \mathbb{N}$ and $\mu < 1$, we obtain
\begin{equation}
\begin{aligned}
    \beta_{k+h+1|k} - \beta_{k+h|k}
    = (\mu - 1) \beta_{k+h|k} - ( \|e_{k+h|k}\|_{\Phi_k}^2- \\
    \theta \|x_{k+h|k}\|_{\Phi_k}^2 ) \le -( \|e_{k+h|k}\|_{\Phi_k}^2 - \theta \|x_{k+h|k}\|_{\Phi_k}^2 ),
\end{aligned}
\label{beta2Phi}
\end{equation}
for all $h > 0$. Based on~\eqref{model2} and~\eqref{beta2Phi}, the left-hand side of the first condition in~\eqref{condition_1} can be reformulated as
\begin{equation}
\begin{aligned}
    &V_{k+h+1|k} - V_{k+h|k}
    = ( \|x_{k+h+1|k}\|_{P_k}^2 - \|x_{k+h|k}\|_{P_k}^2 )+ \\
    &( \|x_{k+h|k}\|_{P_{k,\tau_1}}^2 \!\!\!\!- \|x_{k+h-\tau_1|k}\|_{P_{k,\tau_1}}^2)\! +\! \beta_{k+h+1|k} \!- \!\beta_{k+h|k} \\
    &\le ( \|x_{k+h+1|k}\|_{P_k}^2 - \|x_{k+h|k}\|_{P_k}^2 )
    +    ( \|x_{k+h|k}\|_{P_{k,\tau_1}}^2-\\ 
    & \|x_{k+h-\tau_1|k}\|^2_{P_{k,\tau_1}}) -( \|e_{k+h|k}\|_{\Phi_k}^2 - \theta \|x_{k+h|k}\|_{\Phi_k}^2 ),
\end{aligned}
\label{proof1_eq2}
\end{equation}
where the term $\beta_{k+h|k}$ introduced in~\eqref{quadratic_function_2} incorporates the event-triggering parameter matrix $\Phi_k$ through~\eqref{beta2Phi}. Substituting~\eqref{proof1_eq2} into~\eqref{condition_1}, the inequality~\eqref{condition_1} can be equivalently rewritten as
\begin{equation}
\begin{aligned}
    (\|x_{k+h+1|k}\|_{P_k}^2 - \|x_{k+h|k}\|_{P_k}^2)
    + (\|x_{k+h|k}\|_{P_{k,\tau_1}}^2 -\\ \|x_{k+h-\tau_1|k}\|_{P_{k,\tau_1}}^2) 
 - ( \|e_{k+h|k}\|_{\Phi_k}^2 - \theta \|x_{k+h|k}\|_{\Phi_k}^2 )
    +\\ \|x_{k+h|k}\|_Q^2 + \|u_{k+h|k}\|_R^2 - \varphi \|\omega_{k+h|k}\|_2^2 \le 0.
\end{aligned}
\label{proof1_eq3}
\end{equation}
Using the system dynamics in~\eqref{model2} and the saturated input in~\eqref{saturation_ET}, we divide~\eqref{proof1_eq3} into three parts for separate analysis as following\begin{subequations}
\begin{align}
    & \quad \|x_{k+h+1|h}\|^2_{P_k} + \|x_{k+h|h}\|^2_{P_{k,\tau_1}} \notag \\
    &= \| (A^v + B^v \mathscr{M}_k) x_{k+h|h} 
    + \tilde{A}^v_1 x_{k+h - \tau_1|h}+ \notag \\
    &\quad  B^v \mathscr{M}_k e_{k+h|h} + D \omega_{k+h|h} \|^2_{P_k} 
    + \|x_{k+h|h}\|^2_{P_{k,\tau_1}} \notag \\
    &= z_{k+h|h}^\top \Lambda_1^\top \Upsilon_1 \Lambda_1 z_{k+h|h},
    \label{block_1} \\[0.5em]
    & \quad - ( \|x_{k+h}\|^2_{P_k} + \|x_{k+h - \tau_1}\|^2_{P_{k,\tau_1}} ) -\notag \\
    & \quad ( \|e_{k+h}\|^2_{\Phi_k} - \theta \|x_{k+h}\|^2_{\Phi_k} )
    - \varphi \|\omega_{k+h}\|^2_2 \notag \\
    &= z_{k+h|h}^\top \Upsilon_2 z_{k+h|h},
    \label{block_2} \\[0.5em]
    & \quad \|x_{k+h}\|^2_Q 
    + \|\mathscr{M}_k(e_{k+h} + x_{k+h})\|^2_R \notag \\
    &= z_{k+h|h}^\top \Lambda_2^\top \Lambda_2 z_{k+h|h},
    \label{block_3}
\end{align}
\end{subequations}
where 
\begin{align}
    z_{k+h|h} &= 
    \begin{bmatrix}
        x_{k+h|h}^\top,\; x_{k+h - \tau_1|h}^\top,\; e_{k+h|h}^\top,\; \omega_{k+h|h}^\top
    \end{bmatrix}^\top, \notag \\
    \Lambda_1 &=
    \begin{bmatrix}
        A^v + B^v \mathscr{M}_k & \tilde{A}^v_1 & B^v \mathscr{M}_k & D \\
        I & 0 & 0 & 0
    \end{bmatrix}, \notag \\
    \Lambda_2 &=
    \begin{bmatrix}
        Q^{0.5} & 0 & 0 & 0 \\
        R^{0.5} \mathscr{M}_k & 0 & R^{0.5} \mathscr{M}_k & 0
    \end{bmatrix}, \notag \\
    \Upsilon_1 &= \mathrm{blkdiag}\{ P_k,\; P_{k,\tau_1} \}, \notag \\
    \Upsilon_2 &= \mathrm{blkdiag}\{ -P_k + \theta \Phi_k,\; -P_{k,\tau_1},\; -\Phi_k,\; -\varphi I \}, \notag \\
    \mathscr{M}_k &= E_\eta F_k + E_\eta^{-} H_k, \quad v = 1, \ldots, L. \label{Mk}
\end{align}
The variable matrices $A(k)$, $\tilde{A}_1(k)$, $B(k)$, and $\sum_{\eta = 1}^{2^{n_u}} \varrho_\eta \mathscr{M}_k$ are substituted by their vertices $A^v$, $\tilde{A}_1^v$, $B^v$, and $\mathscr{M}_k$, respectively. With equations in~\eqref{block_1}, \eqref{block_2}, and~\eqref{block_3}, we can convert~\eqref{proof1_eq3} into
\begin{equation}
\begin{aligned}
    z_k^\top \Lambda_1^\top \Upsilon_1 \Lambda_1 z_k 
    + z_k^\top \Upsilon_2 z_k 
    + z_k^\top \Lambda_2^\top \Lambda_2 z_k 
    &= \\z_k^\top ( \Lambda_1^\top \Upsilon_1 \Lambda_1 
    + \Upsilon_2 
    + \Lambda_2^\top \Lambda_2 ) z_k 
    &\le 0,
\end{aligned}
\end{equation}
which is equivalent to
\begin{equation}
\begin{aligned}
 \Lambda_1^\top \Upsilon_1 \Lambda_1 
    + \Upsilon_2 
    + \Lambda_2^\top \Lambda_2  \le 0. \label{proof1_eq6}
\end{aligned}
\end{equation}
Applying the Schur complement, {\eqref{proof1_eq6}} results in
\begin{equation}
    \begin{bmatrix}
        \Xi_1 & \ast \\
        \Xi_2 & \Xi_3\\
    \end{bmatrix}\ge 0, \label{proof1_eq7}
\end{equation}
 where
\begin{align*}
    \Xi_1 &= \mathrm{blkdiag} \left\{
       P_k,\, P_{k,\tau_1},\, \Phi_k,\, \varphi I
    \right\}, \\
    \Xi_2 &= \begin{bmatrix}
        A^v + B^v \mathscr{M}_k & \tilde{A}^v_1 & B^v \mathscr{M}_k & D \\
        I                                   & 0                   & 0                       & 0 \\
        Q^{0.5}                             & 0                   & 0                       & 0 \\
        R^{0.5} \mathscr{M}_k               & 0                   & R^{0.5} \mathscr{M}_k   & 0 \\
        \theta^{0.5} I                      & 0                   & 0                       & 0
    \end{bmatrix}, \\
    \Xi_3 &= \mathrm{blkdiag} \left\{
       P_k^{-1},\, P_{k,\tau_1}^{-1},\, I,\, I,\, \Phi_k^{-1}
    \right\}.
\end{align*}
Pre-multiplying \eqref{proof1_eq7} by 
$\mathrm{blkdiag}\{ 
\gamma_k^{-0.5} W_k^\top, 
\gamma_k^{-0.5} W_{k,\tau_1}^\top,\allowbreak\ 
\gamma_k^{-0.5} W_k^\top,
\gamma_k^{0.5} I, \dots,
\gamma_k^{0.5} I 
\}$
and post-multiplying it by its transpose, the inequality in~\eqref{proof1_eq7} is transformed as
\begin{equation}
    \begin{bmatrix}
        \Xi_4 & \ast \\
        \Xi_5 & \gamma_k\Xi_3\\
    \end{bmatrix} \ge 0,\label{proof1_eq8}
\end{equation}
 where
\begin{align*}
    \Xi_4 &= \mathrm{blkdiag} \big\{ 
        \gamma_k^{-1} W_k^\top P_k W_k,\ 
        \gamma_k^{-1} W_{k,\tau_1}^\top P_{k,\tau_1} W_{k,\tau_1}, \\[-0.5ex]
        &\qquad\qquad 
        \gamma_k^{-1} W_k^\top \Phi_k W_k,\ 
        \varphi \gamma_k I 
    \big\}, \\
    \Xi_5 &= \begin{bmatrix}
        A^v W_k\! + \!B^v \mathscr{M}_k W_k & \!\!\!\tilde{A}^v_1 W_{k,\tau_1} \!\!\!\!\!\!& B^v \mathscr{M}_k W_k & \!\!\!\!\!\gamma_k D \\
        W_k                                         & 0                                & 0                            & 0 \\
        Q^{0.5} W_k                                 & 0                                & 0                            & 0 \\
        R^{0.5} \mathscr{M}_k W_k                   & 0                                & R^{0.5} \mathscr{M}_k W_k    & 0 \\
        \theta^{0.5} W_k                            & 0                                & 0                            & 0
    \end{bmatrix}.
\end{align*}
Notice $\gamma_{k}>0$ and $P_{k}>0$, it follows that $(W_{k} - \gamma_{k}P^{-1}_{k})^\top \gamma_{k}^{-1} P_{k} (W_{k} - \gamma_{k}P^{-1}_{k}) \ge 0$,
which can be equivalently rewritten as
\begin{equation}
	\gamma^{-1}_{k} W^\top _{k} P_{k} W_{k} \ge W_{k} + W^\top _{k} - \gamma_{k} P^{-1}_{k},
	\label{scaling1}
\end{equation}
Similarly, one has
\begin{equation}
	\gamma^{-1}_{k} W^\top_{k,\tau_1} P_{k,\tau_1} W_{k,\tau_1} \ge W_{k,\tau_1} + W^\top_{k,\tau_1} - \gamma_{k} P^{-1}_{k,\tau_1},
	\label{scaling2}
\end{equation}
\begin{equation}
	\gamma^{-1}_{k} W^\top_{k} {\Phi_k} W_{k} \ge W_{k} + W^\top_{k} - \gamma_{k} {\Phi_k}^{-1}.
	\label{scaling3}
\end{equation}
Substituting~\eqref{scaling1}, \eqref{scaling2} and~\eqref{scaling3} into~\eqref{proof1_eq8}, with the variable definition $\gamma_{k}P^{-1}_{k} = Y_{1,k}, \gamma_{k}P^{-1}_{k,\tau_1}= Y_{2,k}, \gamma_{k}{\Phi_k}^{-1}= Y_{3,k}, F_{k}W_{k} = Y_{4,k}$ and $H_{k}W_{k} = Y_{5,k}$, we yields~\eqref{LMI1}. 

By substituting the definition of the \aiping{Lyapunov-Krasovskii}-like function in~\eqref{quadratic_function_2} into condition~\eqref{condition_2}, we obtain
\begin{equation}
    V_{k|k} = \bar{z}^\top _{k|k} \bar{P}_{k} \bar{z}_{k|k} + \beta_{k|k} < \gamma_{k}. \label{proof1_eq9}
\end{equation}
Dividing both sides of~\eqref{proof1_eq9} by $\gamma_{k}^{-1}$, and applying the Schur complement together with the variable definition $\gamma_{k}P^{-1}_{k} = Y_{1,k}$ and $\gamma_{k}P^{-1}_{k,\tau_1} = Y_{2,k}$, we can obtain~\eqref{LMI2}.
\end{proof}

In Lemma~\ref{Lemma_3}, the LMIs~\eqref{LMI1} and~\eqref{LMI2}  satisfy the two conditions in~\eqref{condition_1} and~\eqref{condition_2}, respectively. Thorough this way, we force a upper bound $\gamma_k$ for the cost function $J_\infty(k)$ in~\eqref{cost}. We now introduce the following Lemma~\ref{Lemma_4}, which presents a new LMI for satisfying the constraint in~\eqref{H_LMI_orignal} for express saturation linearly in~\eqref{saturation}.
\begin{lemma}\label{Lemma_4}
    If there exist matrices $W_k$ and $Y_{5,k}$ such that the following LMI holds
    \begin{equation}
        \begin{bmatrix}
            u^2_{\sat} & v_{5,i}^\top \\
            \ast & W_k + W_k^\top - Y_{1,k}
        \end{bmatrix} \ge 0, i = 1, \ldots, n_u, \label{LMI3}
    \end{equation}
where $v_{5,i}^\top$ denotes the $i^{th}$ row of $Y_{5,k}$, then the convex hull condition for saturation in~\eqref{H_LMI_orignal} is satisfied.
\end{lemma}

\begin{proof}
    Based on Lemma~\ref{Lemma_1}, the condition in~\eqref{H_LMI_orignal} can be reformulated as $h_{k,i}^\top (\frac{P_k}{\gamma_k})^{-1} h_{k,i} \le u^2_{\sat}$. Since $Y_{1,k} = \gamma_k P_k^{-1}$, this inequality is equivalent to $ h_{k,i}^\top Y_{1,k} h_{k,i} \le u^2_{\sat}$. By applying the Schur complement, the inequality above is equivalent to the following LMI
    \begin{equation}
        \begin{bmatrix}
            u^2_{\sat} & h_{k,i}^\top \\
            h_{k,i} & Y_{1,k}^{-1}
        \end{bmatrix} \ge 0, \quad i = 1, \ldots, n_u. \label{sat_proof_eq1}
    \end{equation}
Next, pre-multiplying~\eqref{sat_proof_eq1} by $\operatorname{blkdiag}(1, W_k^\top)$ and post-multiplying by its transpose, with $v_{5,i}^\top = h_{k,i}^\top W_k$ and using $W_k^\top Y_{1,k}^{-1} W_k \ge W_k + W_k^\top - Y_{1,k}$, yields the LMI in~\eqref{LMI3}.
\end{proof}
\subsection{Invariant set} \label{Sec3.2}
Invariant sets play a critical role in the analysis of the recursive feasibility. In this section, we derive a condition that ensures the invariance of the closed-loop system in~\eqref{model2}. Let the solution to the optimization problem at time $k$ be denoted by $\{P_{k_t}^*, P_{k_t,\tau_1}^*, \gamma_{k_t}^*\}$. We define the invariant set based on this solution as
\begin{equation}
\Gamma_k \coloneqq \left\{ (\overline{z}, \beta) \mid \overline{z}^\top \overline{P}_{k_t}^* \overline{z} + \beta \le \gamma_{k_t}^* \right\}.
\label{invariank_set}
\end{equation}
We consider the following invariant set condition
\begin{equation}
			\begin{aligned}
&V_{k+h+1|k}-(1-\delta)V_{k+h|k} -\\ 
&\frac{\delta \gamma_{k}}{d^2}\|\omega_{k+h|k}\|^2_2\le0,\,0<\delta<1-\mu,\label{condition_3}
			\end{aligned}
\end{equation}
where $\delta$ is a constant scalar. The invariant set condition~\eqref{condition_3} together with~\eqref{condition_2} ensures that the set $\Gamma$ is an invariant set for the closed-loop system~\eqref{model2}. 
{To clarify the invariant set condition, we} consider the case $h = 0$ in~\eqref{condition_3}, which yields
\begin{equation}
\begin{aligned}
    V_{k+1|k} - (1 - \delta) V_{k|k} - 
    \frac{\delta\gamma_{k}}{d^2} \|\omega_{k|k}\|_2^2 \le 0. \label{Invariank_set_eq1}
\end{aligned}
\end{equation}
By combining~\eqref{condition_2} and~\eqref{disturbance_bound}, we can further derive{
\begin{equation}
\begin{aligned}
    V_{k+1|k} \le (1 - \delta) \gamma_k +\frac{\delta\gamma_k}{d^2} d^2 \le \gamma_k, \label{Invariank_set_eq2}
\end{aligned}
\end{equation}}which confirms that $V_{k+1|k} \le \gamma_k$ holds. By mathematical induction, we conclude that $V_{k+h|k} \le \gamma_k$ for all $h \ge 0$. In view of the definition of $V_{k+h|k}$ in~\eqref{quadratic_function_2}, this implies that the set $\Gamma$ is invariant for the closed-loop system~\eqref{model2}. Next, we derive an LMI condition to guarantee~\eqref{condition_3}, which is summarized in the following lemma.
\begin{lemma}\label{Lemma_5}
Consider the closed-loop system in~\eqref{model2}. Let $Q$, $R$, and $\varphi$ denote the matrices and scalar in~\eqref{cost_J}, and let  $\delta<1-\mu$ and $\theta$ be given positive scalar. If there exist matrices $W_k$, $W_{k,\tau_1}$, $Y_{4,k}$, $Y_{5,k}$, positive definite matrices $Y_{1,k}$, $Y_{2,k}$, $Y_{3,k}$, and a scalar $\gamma_k > 0$, satisfying~\eqref{LMI2} and 
    \begin{equation}
        \begin{bmatrix}
            \Theta_5 & \ast\\
            \Theta_6 & \Theta_7
        \end{bmatrix} \ge 0, \label{LMI4}
    \end{equation}
where
\begin{align*}
     \Theta_5 &= \operatorname{blkdiag} \left\{
        W_k + W_k^\top - Y_{1,k},\;
        W_{k,\tau_1} + W_{k,\tau_1}^\top - Y_{2,k},\right. \nonumber \\
        &\hspace{4em} \left. W_k + W_k^\top - Y_{3,k},\;
        \frac{\delta}{d^2} I
    \right\}, \\
    \Theta_6 &= \begin{bmatrix}
        A^v W_k + B^v \mathscr{N}_k & \tilde{A}_1^v W_{k,\tau_1} & B^v \mathscr{N}_k &  D \\
        W_k & 0 & 0 & 0 \\
        \theta^{0.5} W_k & 0 & 0 & 0
    \end{bmatrix}, \\
    \Theta_7 &= \operatorname{blkdiag} \left\{ Y_{1,k},\; Y_{2,k},\; Y_{3,k} \right\}, \\
\end{align*}with $\mathscr{N}_k$ defined in~\eqref{Nk}{, then} condition~\eqref{condition_3} holds, which is a sufficient condition for the set $\Gamma$ to be an invariant set for the closed-loop system in~\eqref{model2}. 
\end{lemma}
\begin{proof}
Based on~\eqref{quadratic_function_2}, we have
\begin{equation}
    \delta V_{k|k} \le \delta ( \|x_{k+h|k}\|_{P_k}^2 + \|x_{k+h-\tau_1|k}\|_{P_{k,\tau_1}}^2 ).
    \label{proof3_eq2}
\end{equation}
By combining~\eqref{proof1_eq2} and~\eqref{proof3_eq2}, condition~\eqref{condition_3} can be reformulated as 
\begin{equation}
\begin{aligned}
    &( \|x_{k+h+1|k}\|_{P_k}^2 - \|x_{k+h|k}\|_{P_k}^2 ) 
    + ( \|x_{k+h|k}\|_{P_{k,\tau_1}}^2 - \\
    &\|x_{k+h-\tau_1|k}\|_{P_{k,\tau_1}}^2 ) - ( \|e_{k+h|k}\|_{\Phi}^2 - \theta \|x_{k+h|k}\|_{\Phi}^2 ) -\\
   &  \frac{\delta \gamma_k}{d^2} \|\omega_{k+h|k}\|_2^2 
     + \delta ( \|x_{k+h|k}\|_{P_k}^2 + \|x_{k+h-\tau_1|k}\|_{P_{k,\tau_1}}^2 ) \le 0.
\end{aligned}
\label{proof3_eq3}
\end{equation}
We divide~\eqref{proof3_eq3} into two parts, i.e.,~\eqref{block_1} and 
\begin{equation}
\begin{aligned}
     & ( \delta-1)(\|x_{k+h|k}\|_{P_k}^2   - \|x_{k+h-\tau_1|k}\|_{P_{k,\tau_1}}^2)  -\\
     &( \|e_{k+h|k}\|_{\Phi}^2 - \theta \|x_{k+h|k}\|_{\Phi}^2 ) - \frac{\delta \gamma_k}{d^2} \|\omega_{k+h|k}\|_2^2 \\
     &= z_{k+h|h}^\top \Upsilon_3 z_{k+h|h},
\end{aligned}
\label{proof3_eq4}
\end{equation}
where $\Upsilon_3 = \mathrm{blkdiag}\{ ( \delta-1)P_k + \theta \Phi_k,\ ( \delta-1)P_{k,\tau_1},-\Phi_k,\ - \frac{\delta \gamma_k}{d^2} I \}$.
Applying the Schur complement, \eqref{proof3_eq3} is equivalent to
\begin{equation}
    \begin{bmatrix}
        \Xi_6 & \ast \\
        \Xi_7 & \Xi_8\\
    \end{bmatrix}\ge 0, \label{proof3_eq5}
\end{equation}
 where
\begin{align*}
    \Xi_6 &= \mathrm{blkdiag} \left\{
       (1-\delta)P_k,\, (1-\delta)P_{k,\tau_1},\, \Phi_k,\, \frac{\delta \gamma_k}{d^2} I
    \right\}, \\
    \Xi_7 &= \begin{bmatrix}
        A^v + B^v \mathscr{M}_k & \tilde{A}^v_1 & B^v \mathscr{M}_k & D \\
        I                                   & 0                   & 0                       & 0 \\
        \theta^{0.5} I                      & 0                   & 0                       & 0\\
    \end{bmatrix}, \\
    \Xi_8 &= \mathrm{blkdiag} \left\{
       P_k^{-1},\, P_{k,\tau_1}^{-1},\, \, \Phi_k^{-1}
    \right\},
\end{align*}
with $\mathscr{M}_k$ defined in~\eqref{Mk}. Pre-multiplying \eqref{proof3_eq5} by 
$\mathrm{blkdiag}\{ 
\gamma_k^{-0.5} W_k^\top, 
\gamma_k^{-0.5} W_{k,\tau_1}^\top,
\gamma_k^{-0.5} W_k^\top,
\gamma_k^{-0.5} I, 
\gamma_k^{0.5} I ,\allowbreak\
\gamma_k^{0.5} I ,\allowbreak\
\gamma_k^{0.5} I 
\}$
and post-multiplying it by its transpose, and making use of~\eqref{scaling1},~\eqref{scaling2}, and~\eqref{scaling3}, we obtain the LMI condition~\eqref{LMI4}, which guarantees that~\eqref{condition_3} is satisfied. As in the analysis from~\eqref{Invariank_set_eq1} to~\eqref{Invariank_set_eq2}, this implies that $\Gamma$ is an invariant set defined by the solution for the closed-loop system in~\eqref{model2}.
\end{proof}

\subsection{Adaptive ETMPC algorithm} \label{Sec3.3}
In the previous section, we propose Lemmas~\ref{Lemma_3},~\ref{Lemma_4} and~\ref{Lemma_5}, which focus on 1) guaranteeing an upper bound for the infinite-horizon cost function $J_\infty(k)$, 2) reformulating the nonlinear saturated input as a convex combination of the vertices defined in~\eqref{saturation} and 3) ensuring the invariant set $\Gamma$, {respectively}. These results provide the foundation for our MPC design. Based on the above, the complete ETMPC optimization problem can be formulated as 
\begin{equation}
    \begin{aligned}
        &\min_{W_k, W_{k,\tau_1}, Y_{1,k}>0, Y_{2,k}>0, Y_{3,k}>0,Y_{4,k}, Y_{5,k}} \;  \gamma_k, \\
        &\text{subject to}~\eqref{LMI1},~\eqref{LMI2},~\eqref{LMI3}\text{ and}~\eqref{LMI4}.
    \end{aligned}\label{minimize_gamma}
\end{equation}

To clearly present the proposed control framework, we summarize the overall procedure of the adaptive ETMPC {algorithm} in Algorithm~\ref{ETMPC}. The algorithm iteratively checks the triggering condition. ETMPC solves the LMI-based optimization problem in~\eqref{minimize_gamma} only when necessary, thereby reducing the computational load.
\begin{algorithm} [h]
\caption{{Adaptive} Event-Triggered MPC}
\label{ETMPC}
\begin{algorithmic}[1]
\STATE \textbf{Initialization:}
Set the parameters $\{A^v, \tilde{A}_1^v, \dots, B^v\}_{v=1}^L$ in~\eqref{convex_hull}, weighting matrices $Q, R$ and scalar $\varphi$ in~\eqref{cost_J}.
Set the event-triggering parameters $\mu, \theta, \varepsilon$ satisfying~\eqref{para_constraints} and a constant scalar $0<\delta<1-\mu$.
Initialize the state $x_0$ and the internal adaptive variable $\beta_0 \ge 0$.

\FOR{$k = 0, 1, 2, \dots$}
    \STATE Obtain the current system state $x_k$ and calculate the error vector $e_k = x_{k_t} - x_k$.
        \IF{ $\varepsilon ( \|e_{k}\|_{\Phi_{k_t}} - \theta \|x_{k}\|_{\Phi_{k_t}} ) > \beta_{k}$ \OR $k = 0$ }
        \STATE Set the new trigger instant $k_t \leftarrow k$.
        \STATE Solve the optimization problem (\ref{minimize_gamma}) based on the current state $x_{k_t}$ to obtain the optimal solution set $\{\gamma_{k_t}^*, W_{k_t}^*, W_{k_t,\tau_1}^*, V_{1,k_t}^*, \dots, V_{5,k_t}^*\}$.
        \STATE Compute the controller $F_{k_t} \leftarrow V_{4,k_t}^* (W_{k_t}^*)^{-1}$ and the event-triggering matrix $\Phi_{k_t} \leftarrow \gamma_{k_t}^* (V_{3,k_t}^*)^{-1}$

        \STATE Compute the new control input $u_k \leftarrow F_{k_t} x_{k_t}$.
    \ELSE
        \STATE Maintain the previous control $u_k \leftarrow u_{k-1}$ (ZOH).
    \ENDIF
    \STATE Apply the control input $\sigma(u_k)$ to the plant (\ref{model1}) to obtain the next state $x_{k+1}$.
    \STATE Update the internal adaptive variable $\beta_{k+1}$ in~\eqref{internal_adaptive_variable}.
\ENDFOR
\end{algorithmic}
\end{algorithm}

We present the following theorem to establish that the proposed adaptive ETMPC algorithm guarantees recursive feasibility and satisfies the mean-square ISS. \begin{theorem}
Consider the system~\eqref{model1}. If the optimization problem~\eqref{minimize_gamma}, constrained by the LMIs in~\eqref{LMI1},~\eqref{LMI2},~\eqref{LMI3} and ~\eqref{LMI4}, admits a feasible solution at any time step $k$, then the closed-loop system~\eqref{model2} is guaranteed to be mean-square ISS with respect to the disturbance $\omega_k$.
\end{theorem}
\begin{proof}
    The proof consists of two steps.
    
Step 1 (Recursive Feasibility): 
Assume that the optimization problem~\eqref{minimize_gamma}, subject to the LMIs~\eqref{LMI1},~\eqref{LMI2},~\eqref{LMI3} and ~\eqref{LMI4}, is feasible at the triggering instant $k_t$. Let the corresponding optimal solution be denoted by 
$\{\gamma^*_{k_t}, W^*_{k_t}, W^*_{k_t,\tau_1}, V^*_{1,k_t}, V^*_{2,k_t}, V^*_{3,k_t}, V^*_{4,k_t}, V^*_{5,k_t}\}$.
At the subsequent triggering instant $k_{t+1}$, it is observed that only~\eqref{LMI2} explicitly depends on the updated state $x_{k_{t+1}}$. Therefore, to establish recursive feasibility, it suffices to verify that~\eqref{LMI2} remains satisfied at $k_{t+1}$. Specifically, at $k_{t+1}$,~\eqref{LMI2} can be rewritten as
	$x_{k_{t+1}}^\top V_{1,k_{t+1}}^{-1} x_{k_{t+1}} 
    + \sum_{\rho=1}^{\tau_1} x_{k_{t+1}-\rho}^\top V_{2,k_{t+1}}^{-1} x_{k_{t+1}-\rho}
    + \gamma_{k_{t+1}}^{-1} \beta_{k_{t+1}} < 1.$ 
We now consider a candidate solution at $k_{t+1}$ given by
\begin{equation}
    \{\gamma_{k_{t+1}}, \dots, V_{5,k_{t+1}}\} := \{\gamma^*_{k_t}, \dots, V^*_{5,k_t}\}. \label{proof4_eq1}
\end{equation}
Under this assignment, inequality~\eqref{LMI2} becomes
\begin{equation}
    \begin{aligned}
    	&x_{k_{t+1}}^\top {V^*_{1,k_t}}^{-1} x_{k_{t+1}}   + \\
    &\sum_{\rho=1}^{\tau_1} x_{k_{t+1}-\rho}^\top {V^*_{2,k_t}}^{-1} x_{k_{t+1}-\rho}
    + {\gamma^*_{k_t}}^{-1} \beta_{k_{t+1}} < 1. \label{proof4_eq2}
\end{aligned}
\end{equation}

Furthermore, the feasibility of~\eqref{LMI2} and~\eqref{LMI4} guarantees that the set $\Gamma$ is invariant for the closed-loop system~\eqref{model2}. This implies
    $x_{k_{t+1}|k_t}^\top P_{k_t} x_{k_{t+1}|k_t}
+ \sum_{\rho=1}^{\tau_1} x_{k_{t+1}-\rho|k_t}^\top P_{k_t,\tau_1} x_{k_{t+1}-\rho|k_t}
+ \beta_{k_t} < \gamma^*_{k_t}.$
Using the substitutions $\gamma_k P_k^{-1} = Y_{1,k}$ and $\gamma_k P_{k,\tau_1}^{-1} = Y_{2,k}$ from Lemma~\ref{Lemma_1}, this inequality becomes
\begin{equation}
    \begin{aligned}
            &x_{k_{t+1}|k_t}^\top {V^*_{1,k_t}}^{-1} x_{k_{t+1}|k_t}+ \\
&\sum_{\rho=1}^{\tau_1} x_{k_{t+1}-\rho|k_t}^\top {V^*_{2,k_t}}^{-1} x_{k_{t+1}-\rho|k_t}
+ {\gamma^*_{k_t}}^{-1} \beta_{k_t} < 1. \label{proof4_eq3}
    \end{aligned}
\end{equation}
Note that the predicted states $x_{k_{t+1}|k_t}$ and $x_{k_{t+1}-\rho|k_t}$ $(\rho=1,\dots,\tau_1)$ are computed based on the dynamics in~\eqref{model2}, which are affected by the disturbance $\omega_k$, the time-varying parameter $\alpha(k)$ from~\eqref{para_vary}, the uncertain parameter $\varrho_\eta$, and the uncertain error $e_k$. Since inequality~\eqref{proof4_eq3} holds for all admissible values of $\omega_k$, $\alpha(k)$, $\varrho_\eta$, and $e_k$, it implies that the actual state measurements $x_{k_{t+1}}$ and $x_{k_{t+1}-\rho}$ will also satisfy~\eqref{proof4_eq3}. Therefore, the inequality~\eqref{proof4_eq2} holds, meaning that~\eqref{LMI2} remains feasible at time $k_{t+1}$.
By the same reasoning, recursive feasibility is preserved at instants $k_{t+2}, k_{t+3}, \ldots$. This completes the proof that the LMIs~\eqref{LMI1},~\eqref{LMI2},~\eqref{LMI3}, and~\eqref{LMI4} are recursively feasible under the proposed adaptive ETMPC scheme.

Step 2 (Mean-Square ISS): The feasibility of \eqref{LMI1}, \eqref{LMI2}, \eqref{LMI3}, and \eqref{LMI4} further leads to inequality~\eqref{proof1_eq3}. In the following, we prove the mean-square ISS based on~\eqref{proof1_eq3}. 
For a constant $\mu$ in~\eqref{para_constraints}, there exists a parameter $a \in (0, 1 - \mu)$ such that $\mu - 1 + a < 0$. Combining this condition with~\eqref{proof1_eq3}, we obtain
\begin{equation}
\begin{aligned}
    &(\|x_{k+1|k}\|_{P_k}^2 - \|x_{k|k}\|_{P_k}^2)+ (\|x_{k|k}\|_{P_{k,\tau_1}}^2 -\\
    &\|x_{k-\tau_1|k}\|_{P_{k,\tau_1}}^2) - ( \|e_{k|k}\|_{\Phi}^2 - \theta \|x_{k|k}\|_{\Phi}^2 )   + \|x_{k|k}\|_Q^2\\ 
    &+ \|u_{k|k}\|_R^2 - \varphi \|\omega_{k|k}\|^2 +(\mu - 1 + a)\beta_k\le 0.
\end{aligned}
\label{proof4_eq4}
\end{equation}
Note that $(\mu - 1) \beta_k - ( \|e_{k|k}\|_{\Phi}^2 - \theta \|x_{k|k}\|_{\Phi}^2 ) - \beta_k = \beta_{k+1} - \beta_k,$ which leads to the following inequality
\begin{equation}
\begin{aligned}
  V_{k+1|k} - V_{k|k} +\|x_{k|k}\|_Q^2 + \|u_{k|k}\|_R^2\\ - \varphi \|\omega_{k|k}\|^2_2 +a\beta_k \le 0.
\end{aligned}
\label{proof4_eq5}
\end{equation}
Following a similar argument as from~\eqref{proof4_eq3} to~\eqref{proof4_eq2}, we extend~\eqref{proof4_eq5} to the actual system trajectory
\begin{equation}
\begin{aligned}
  V_{k+1} - V_{k} +\|x_{k}\|_Q^2 + \|u_{k}\|_R^2\\ - \varphi \|\omega_{k}\|^2_2 +a\beta_k \le 0.
\end{aligned}
\label{proof4_eq6}
\end{equation}
We define the augmented variable $\xi_k = \begin{bmatrix} \overline{z}_k^\top & \beta_k^{0.5} \end{bmatrix}^\top$, and $\varrho_1 = \min\{ \lambda_{\min}(Q), a \}$, where $\lambda_{\min}(\cdot)$ and $\lambda_{\max}(\cdot)$ denote the minimum and maximum eigenvalues of a matrix, respectively. Then, from~\eqref{proof4_eq6}, it follows that
\begin{equation}
\begin{aligned}
  V_{k+1} - V_{k} \le -\varrho_1 \|\xi_k\|_2^2 +  \varphi \|\omega_{k}\|_2^2.
\end{aligned}
\label{proof4_eq7}
\end{equation}

In addition, by the definition of $V_k$ in~\eqref{quadratic_function_2}, we have
\begin{equation}
\begin{aligned}
\lambda_{\min}(P_{k},P_{k,\tau_1})\|\overline{z}_{k}\|_2^2 + \beta_{k}
\leq V_{k} \leq\\ \lambda_{\max}(P_{k},P_{k,\tau_1})\|\overline{z}_{k}\|_2^2 + \beta_{k} .
\end{aligned}\label{proof4_eq8}
\end{equation}
This can be equivalently expressed as
\begin{equation}
    \varrho_2 \|\xi_k\|_2^2 \le V_k \le \varrho_3 \|\xi_k\|_2^2,
    \label{proof4_eq9}
\end{equation}
where $\varrho_2 = \min\{\lambda_{\min}(P_{k},P_{k,\tau_1}), 1\}$ and $\varrho_3 = \max\{\lambda_{\max}(P_{k},P_{k,\tau_1}), 1\}$. According to Lemma~$6$ in~\cite{song2018n}, inequalities~\eqref{proof4_eq7} and~\eqref{proof4_eq9} imply that
\[
\|\xi_k\|_2^2 \le \psi_1(\|\xi_0\|_2^2, k) + \psi_2(\|\omega_k\|_\infty^2),
\]
where $\psi_1 \in \mathcal{KL}$ and $\psi_2 \in \mathcal{K}$. Finally, note that $\|\overline{z}_k\|_2^2 \le \|\xi_k\|_2^2$ and $\|\xi_0\|_2^2 = \|\overline{z}_0\|_2^2 + \beta_0$. This directly implies that the closed-loop system in~\eqref{model2} satisfies the condition in Definition~\ref{Definition_1}, thus completing the proof of mean-square ISS.
\end{proof}
\vspace{-1.5ex}
\section{Simulation Example} \label{Sec4}
In this section, the proposed adaptive ETMPC {algorithm} is validated on an industrial electric heater used for thermal processing. The heater system has been previously studied in~\cite{chu1993time}. The process is partitioned into five zones, each governed by an individual control input. The objective is to regulate the temperature of the industrial heater to a desired setpoint in each zone (see Fig.~\ref{fig:heater_system}). 

\begin{figure} [htb]
	\centering
	\includegraphics[width=0.4\textwidth]{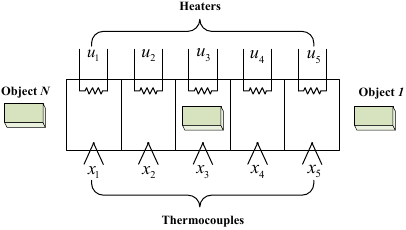}  
	\caption{Schematic diagram of the industrial electric heater.}
	\label{fig:heater_system}
\end{figure}

The system model follows the structure described in~\cite{chu1993time}, and is given by
\begin{equation}
    x_{k+1} = A(k) x_k + \tilde{A}_1(k) x_{k-\tau_1} + B(k) \sigma(u_k) + D \omega_k, \label{simulation_1}
\end{equation}
where the state and control input vectors are defined as
$x = \begin{bmatrix} \Delta T_1 & \Delta T_2 & \Delta T_3 & \Delta T_4 & \Delta T_5 \end{bmatrix}^\top$, 
$u = \begin{bmatrix} \Delta u_1 & \Delta u_2 & \Delta u_3 & \Delta u_4 & \Delta u_5 \end{bmatrix}^\top,$
with $\Delta T_i = T_i - \overline{T}_i$ and $\Delta u_i = u_i - \overline{u}_i$. Here, $T_i$ and $u_i$ ($i = 1, \ldots, 5$) denote the temperature and electric current in the $i$th zone, respectively. The values $\overline{T}_i$ and $\overline{u}_i$ are the steady-state operating points, given by
$\overline{T} = \begin{bmatrix} 696 & 747 & 774 & 774 & 731 \end{bmatrix}^\top $, 
$\overline{u} = \begin{bmatrix} 7.0 & 4.0 & 2.5 & 4.0 & 3.5 \end{bmatrix}^\top.$
The system has a time delay of $\tau_1 = 2$, and the control inputs are constrained as $-0.4 \le \Delta u_i \le 0.4$ for all $i = 1, \ldots, 5$. The nominal system matrices identified in~\cite{chu1993time}.

{To introduce parameter variation, we use the first vertex $A^1, \tilde{A}^1$ and $B^1$ from the model settings in~\cite{chu1993time}, and define the second vertex as $A^2 = 0.9 A^1$, $\tilde{A}^2 = 0.9 \tilde{A}^1$, and $B^2 = B^1$.} The disturbance matrix is $D = 0.1 I$, and the additive disturbance is $\omega_k = 0.01[ \sin(k), \ldots, \sin(k)]^\top$.
The weighting matrices for the cost function are selected as
$Q = 0.01 \times \mathrm{diag}(0.25, 0.11, 1, 1, 0.5)$,  $R = 0.001 I$, 
with the remaining parameters set as $\varphi = 10$, $\theta = 0.1$, $\varepsilon = 1.12$, $\mu = 0.9$, $\delta = 0.09$, and $d^2 = 0.0018$, $\beta_0=10$.
The control objective is to steer the system from the initial condition $x_0 = [1.2\; 0.9\; 1.2\; 0.9\; 1.2]^\top$ to the origin, i.e., to drive the temperatures in all zones to their respective operating points, such that $T_i \to \overline{T}_i$.

\begin{figure}
    \centering
    \includegraphics[width=0.4\textwidth]{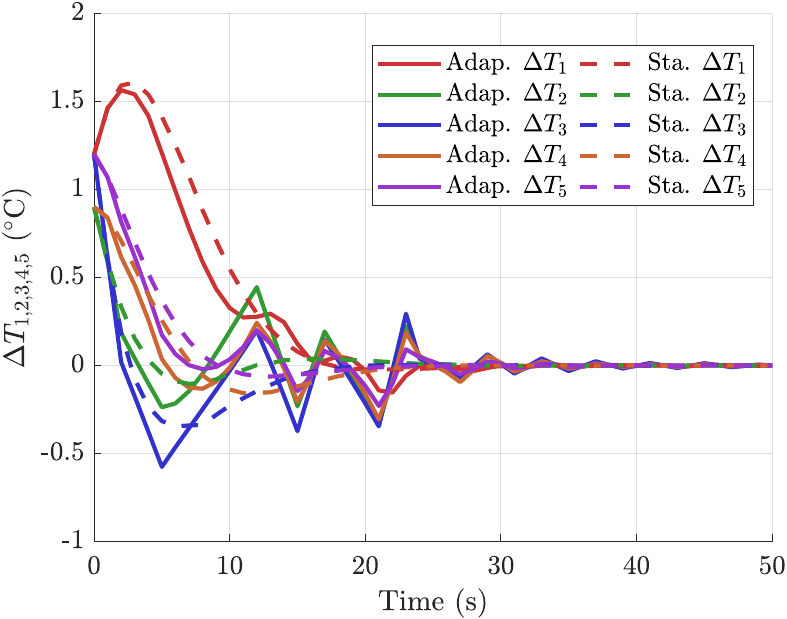}  
    \caption{State trajectories of the industrial heater under the proposed adaptive ETMPC~\eqref{adaptiveETM} and the static ETMPC~\eqref{staticETM}.}
    \label{fig:states}
\end{figure}

\begin{figure}
    \centering
    \includegraphics[width=0.4\textwidth]{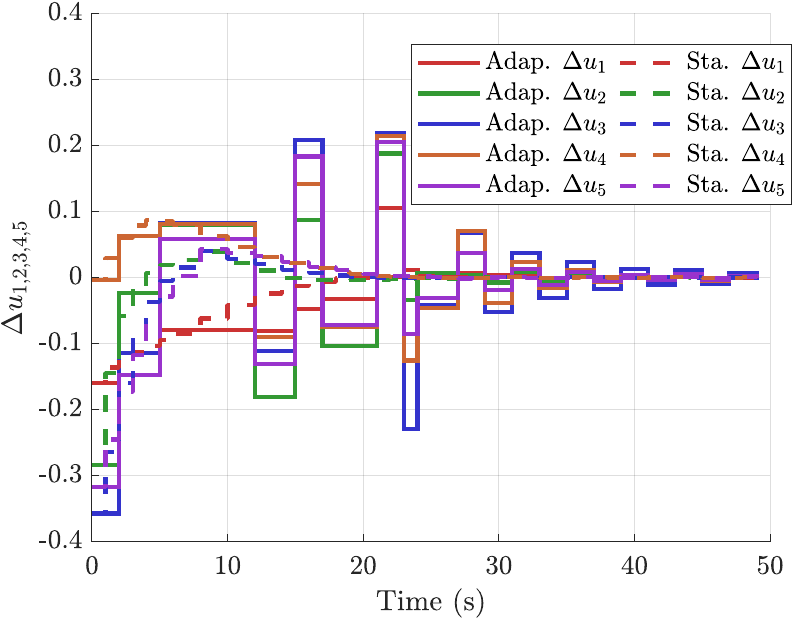}  
    \caption{Control input signals under adaptive ETMPC~\eqref{adaptiveETM} and static ETMPC~\eqref{staticETM}.}
    \label{fig:inputs}
\end{figure}

To evaluate the performance of the proposed adaptive ETMPC in~~\eqref{adaptiveETM} against the conventional static ETMPC in~\eqref{staticETM}, simulation experiments {are} conducted on the five-zone industrial heater system. The results are presented in Figs.~\ref{fig:states}–\ref{fig:trigger_beta_gamma} and Table~\ref{tab:trigger_stats}.
Fig.~\ref{fig:states} shows the evolution of the temperature deviation variables $\Delta T_1$ to $\Delta T_5$ under both control strategies. The steady-state criterion is defined by
\[
k_s = \min \left\{ k \in [0, k_f] \;\middle|\; \max_{j \ge k, i=1,\dots,n_x} |x_{i,j}| \leq \zeta \right\},
\]
where $k_f$ denotes the total simulation time and the threshold is set as $\zeta=0.05 $. The system under the static ETMPC reaches the steady-state at $k_s = 19.00$~s, whereas the adaptive ETMPC reaches it at $k_s = 30.00$~s. The result indicates that, although both controllers eventually drive the system state to equilibrium, the static ETMPC achieves faster convergence in terms of the steady-state criterion. This reflects a trade-off between convergence speed and control frequency. The adaptive ETMPC prioritizes communication efficiency, potentially at the cost of delayed convergence.

The control input trajectories $\Delta u_1$ to $\Delta u_5$ are illustrated in Fig.~\ref{fig:inputs}. Under the adaptive ETMPC, control updates occur less frequently, owing to the adaptive triggering condition in~\eqref{adaptiveETM} and the zero-order hold implementation. Conversely, the static ETMPC generates more frequent control signals. This may increase the computational load and cause excessive actuator switching, which can accelerate wear in practical systems.

\begin{figure}
	\centering
	\includegraphics[width=0.4\textwidth]{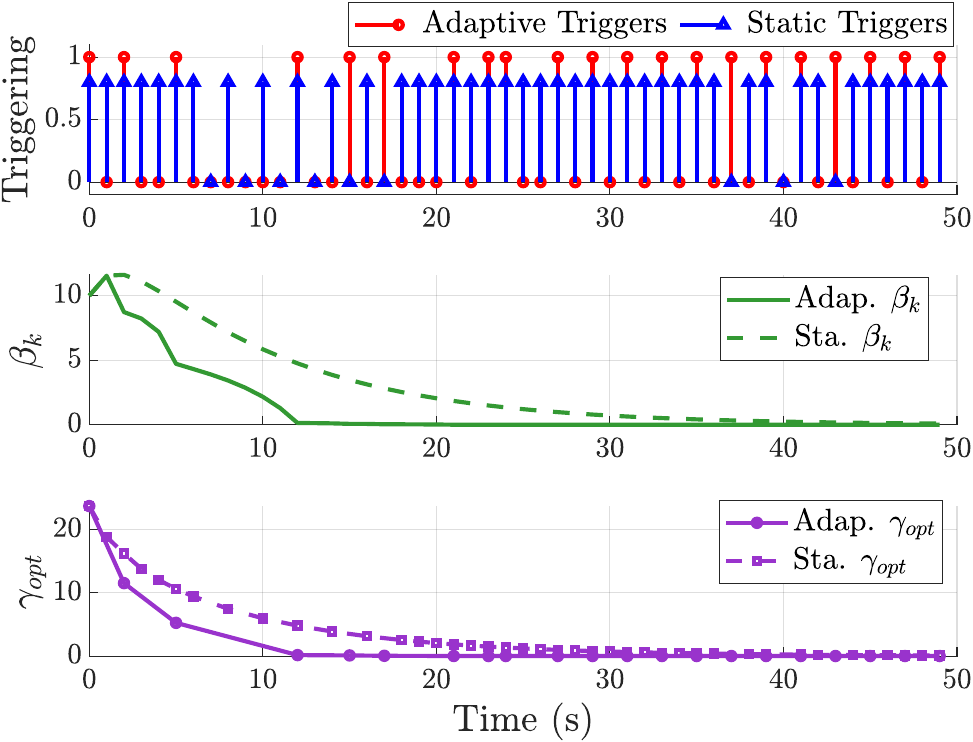}  
	\caption{Illustration of triggering instants (top), internal variable $\beta_k$ (middle), and optimal cost bound $\gamma_{\mathrm{opt}}$ (bottom) under the adaptive ETMPC~\eqref{adaptiveETM} and static ETMPC~\eqref{staticETM}.}
	\label{fig:trigger_beta_gamma}
\end{figure}

\begin{table}
    \centering
    \caption{Event-triggering statistics under $50$ simulation steps}
    \label{tab:trigger_stats}
    \resizebox{\linewidth}{!}{%
    \begin{tabular}{lcc}
        \toprule
        Metric & Adaptive ETMPC~\eqref{adaptiveETM} & Static ETMPC~\eqref{staticETM} \\
        \midrule
        Triggering ratio & $42.00\%$ & $82.00\%$ \\
        Average interval (steps) & $2.38$ & $1.22$ \\
        \bottomrule
    \end{tabular}%
    }
\end{table}

Fig.~\ref{fig:trigger_beta_gamma} further examines the internal behaviors of both schemes. The top subplot depicts the triggering instants over the simulation horizon, where the adaptive ETMPC leads to substantially fewer triggering events than the static ETMPC. The middle subplot shows the evolution of the internal variable $\beta_k$, which decreases adaptively and converges to zero. The bottom subplot presents the evolution of the optimal upper bound $\gamma_{\mathrm{opt}}$ for the infinite-horizon cost function. The adaptive ETMPC achieves a tighter and faster convergence of $\gamma_{\mathrm{opt}}$.
These observations are quantitatively supported by the data in Table~\ref{tab:trigger_stats}. Over the $50$-step simulation, the adaptive ETMPC achieves an event triggering ratio of only $42.00\%$, compared to $82.00\%$ for the static ETMPC. This leads to an average inter-trigger interval of $2.38$ steps, nearly twice that of the static counterpart, thus demonstrating the communication efficiency and reduced computational demand achieved by the proposed adaptive ETMPC {algorithm}.

\begin{figure}
    \centering
    \includegraphics[width=0.4\textwidth]{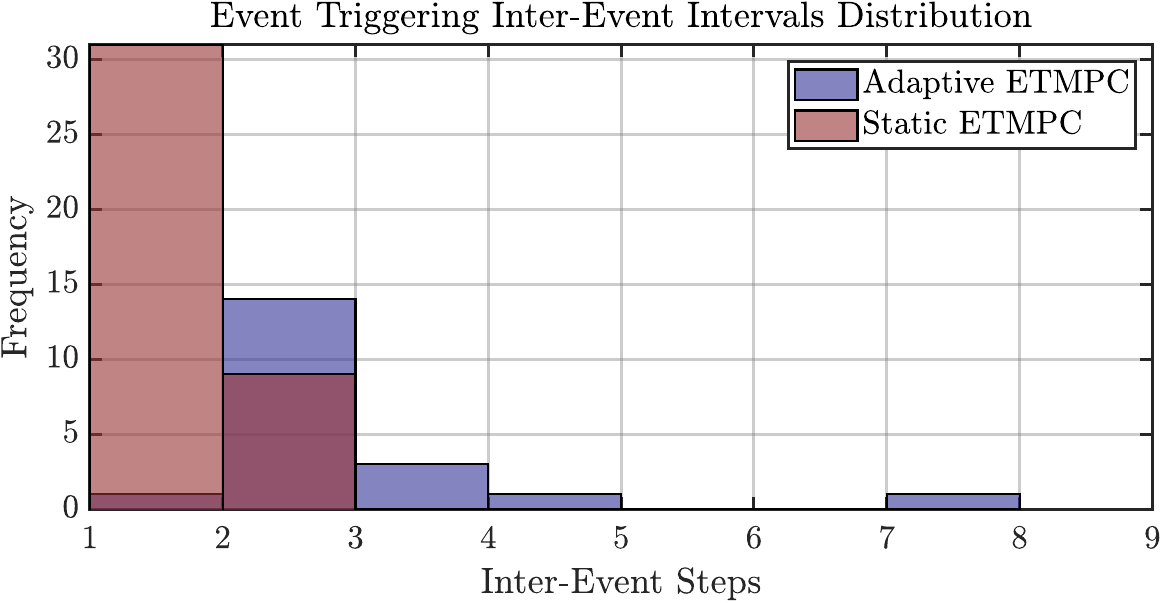}
    \caption{Histogram of inter-event intervals under the adaptive ETMPC~\eqref{adaptiveETM} and static ETMPC~\eqref{staticETM}.}
    \label{fig:histogram}
\end{figure}
\begin{figure}
    \centering
    \includegraphics[width=0.4\textwidth]{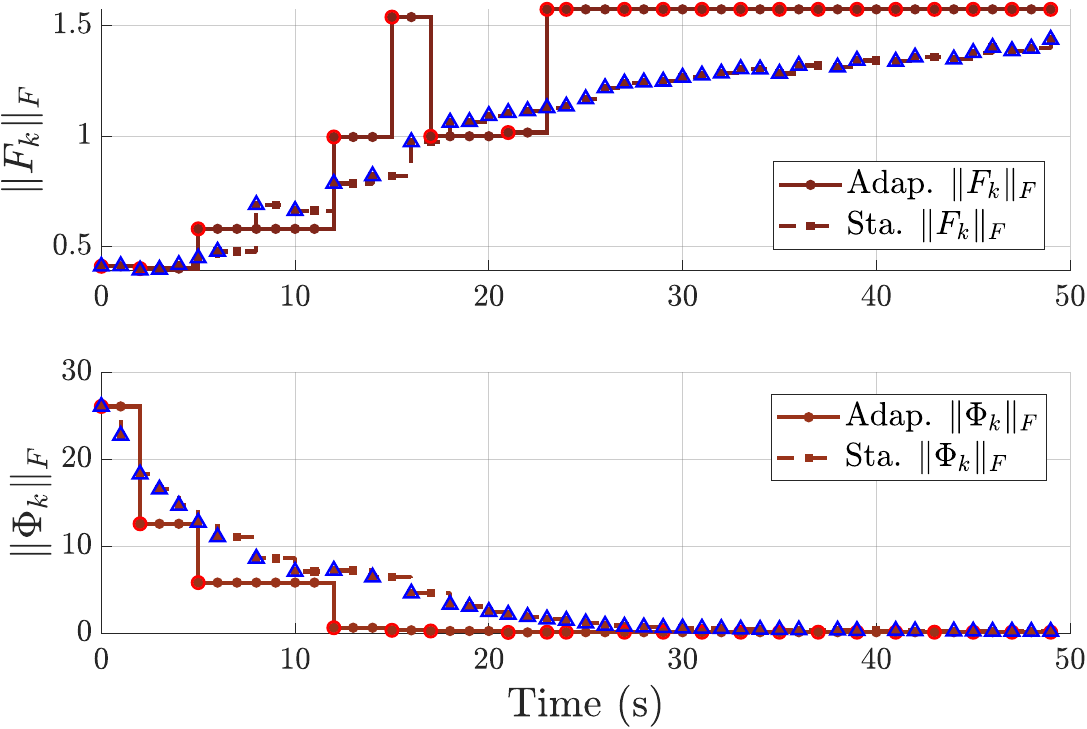}
\caption{Evolution of control gain norm $\|F_k\|_F$ (top) and triggering matrix norm $\|\Phi_k\|_F$ (bottom) for the adaptive ETMPC~\eqref{adaptiveETM} and static ETMPC~\eqref{staticETM}.}
    \label{fig:parameter_evolution}
\end{figure}

Figure~\ref{fig:histogram} presents the distribution of inter-event intervals under both control strategies. While Table~\ref{tab:trigger_stats} summarizes only the average inter-trigger interval, the histogram provides a more informative view of how those intervals are distributed. The static ETMPC exhibits a highly concentrated distribution around short intervals (e.g., 1 or 2 steps), indicating dense and frequent triggering. In contrast, the dynamic ETMPC produces a wider distribution, with a substantial number of longer intervals such as 3, 4, or even 7 steps. This evidences that the dynamic strategy not only reduces the average triggering frequency, but also promotes more frequent occurrences of long silent periods.
Figure~\ref{fig:parameter_evolution} illustrates the evolution of key optimized parameters throughout the control process, where $\|\cdot\|_F$ means the Frobenius norm. Since these parameters are updated only at triggering instants, the plots appear piecewise constant. Under the adaptive ETMPC, both $\|F_k\|_F$ and $\|\Phi_k\|_F$ exhibit more significant variations during the early stages. Notably, $\|F_k\|_F$ becomes stable after $t = 23$ and $\|\Phi_k\|_F$ after $t = 12$, both ahead of the stabilization observed in the static ETMPC. These observations indicate that the proposed method adaptively adjusts both the controller and the triggering mechanism in response to real-time system dynamics.


\bibliographystyle{elsarticle-num}
\bibliography{Ref}

\end{document}